%% file: main.tex
\documentclass[11pt]{article}

\usepackage[margin=1in]{geometry}
\usepackage{amsmath,amssymb,amsthm,mathtools}
\usepackage{enumitem}
\usepackage{booktabs}
\usepackage{float, caption, subcaption}
\usepackage{microtype}
\usepackage{xcolor, xspace}
\usepackage{hyperref}
\usepackage[nameinlink,capitalize,noabbrev]{cleveref}
\usepackage{thmtools}
\usepackage{thm-restate}
\usepackage{algorithm}

\newtheorem{theorem}{Theorem}[section]
\newtheorem{lemma}[theorem]{Lemma}
\newtheorem{proposition}[theorem]{Proposition}
\newtheorem{corollary}[theorem]{Corollary}

\theoremstyle{definition}
\newtheorem{definition}[theorem]{Definition}

\newcommand{\R}{\mathbb{R}}

\newcommand{\arcs}{\text{Arcs}}

\def\final{0}  
\def\iflong{\iffalse}
\ifnum\final=0  
\newcommand{\knote}[1]{{\color{blue}[{Karthik: \bf #1}]\marginpar{\color{blue}*}}}
\newcommand{\todo}[1]{{\color{red}[{TODO: \bf #1}]\marginpar{\color{red}*}}}
\newcommand{\krnote}[1]{{\color{purple}[{Krishna: \bf #1}]\marginpar{\color{purple}*}}}
\else 
\newcommand{\knote}[1]{}
\newcommand{\todo}[1]{}
\newcommand{\krnote}[1]{}
\fi  

\title{Polynomial-time Stable Matching in Network Hypergraphs\thanks{Grainger College of Engineering, University of Illinois, Urbana-Champaign, Email: {\tt\{karthe, kk17\}@illinois.edu}. Supported in part by NSF grant CCF-2402667. }}
\author{Karthekeyan Chandrasekaran \and Krishna Kalathur}
\date{}

\begin{document}
\maketitle

\input{abstract}
\input{intro-1}
\input{stable-matching-and-kernel}

\input{DE-graphs}

\input{network-hypergraphs}
\input{conclusion}

\paragraph{AI Disclosure.} The key connection that drives the algorithm was discovered by chatting with ChatGPT-5.6 Sol Max. The authors verified it independently and have presented the details in their own words. The authors used ChatGPT for editorial polishing and assume full responsibility for all content. 

\paragraph{Acknowledgements.} Karthekeyan thanks Gergely Cs\'{a}ji and Ildik\'{o} Schlotter for preliminary discussions on these topics at the Eml\'ekt\'abla Workshop 2024 \cite{emlektabla-report}. 


\bibliographystyle{abbrv}
\bibliography{references}

\appendix

\end{document}

%% file: abstract.tex
\begin{abstract}
    We show that there exists a polynomial-time algorithm to find a stable matching in network hypergraphic preference systems. 
    The key connection that drives the algorithm was discovered by chatting with ChatGPT-5.6 Sol Max. We verified it independently and present the details in our own words. 
\end{abstract}

%% file: intro-1.tex
\section{Introduction}
Stable matching is a foundational model at the interface of combinatorial optimization and economics. In the classical bipartite setting, the deferred-acceptance algorithm of Gale-Shapley \cite{GS62} proves that a stable matching always exists and can be found in polynomial time. Variants of this algorithm underpin matching mechanisms in markets such as medical residency \cite{Rot84}. 
Bipartite stable matching also has a rich mathematical structure: stable matchings form a distributive lattice, admit compact representations through rotations, several combinatorial algorithms are known, and the polyhedral structure is well-understood \cite{GusfieldIrving1989, Rothblum1992}. This successful combination of applications, algorithms, and structure has made stable matching a central paradigm for matching under preferences. 

The classical bipartite stable matching model is, however, intrinsically pairwise and two-sided. In many settings, the relevant unit of cooperation is a coalition rather than a pair. They naturally lead to a hypergraph where vertices correspond to agents and hyperedges describe coalitions. 
Aharoni and Fleiner \cite{AF03} modeled such settings by a \emph{hypergraphic preference system}. A hypergraphic preference system is a pair $(H, \succ)$, where $H=(V, E)$ is a hypergraph and 
\[
\succ = (\succ_v)_{v\in V}
\]
is a preference profile: for every vertex $v$, the relation $\succ_v$ is a strict ordering of the hyperedges incident to $v$. 
A \emph{matching} is a set of pairwise vertex-disjoint hyperedges. A matching is \emph{stable} if every unchosen hyperedge $e$ contains a vertex that prefers its matched hyperedge over $e$. We assume throughout that $H$ has singleton hyperedges $e_v=\{v\}$ for each $v\in V$ and each $v\in V$ ranks $e_v$ as its least preferred hyperedge, essentially representing the outside option of remaining unmatched. Thus, vertices correspond to agents, hyperedges correspond to feasible coalitions, and a stable matching seeks a collection of feasible coalitions such that no coalition can deviate in a way that benefits all of its members. Stable hypergraph matching is closely connected to core stability in finitely generated non-transferable utility (NTU) games and hedonic coalition games \cite{BanerjeeKonishiSonmez2001, AzizBrandl2012, Woeginger2013, BF16}. 

The move pairwise and two-sided coalitions (i.e., from bipartite scenarios) to arbitrary coalitions changes the existential picture dramatically. 
A stable matching need not exist for an arbitrary hypergraphic preference system, and deciding existence is NP-complete \cite{NgHirschberg1991}. 
Nevertheless, Aharoni and Fleiner \cite{AF03} applied Scarf's lemma to conclude that every such system has a \emph{fractional} stable matching. Moreover, their fractional stable matching may be chosen as an extreme point of the fractional matching polytope 
\[
P(H) = \{x\in \R^E_{\ge 0}: A_H x =1\},
\]
where $A_H$ is the vertex-hyperedge incidence matrix of $H$. In particular, if $P(H)$ is an integral polyhedron, then their results yield an integral stable matching \cite{AF03, BF16}. 
Bir\'{o} and Fleiner \cite{BF16} formalized the interpretation of Scarf solutions as fractional core elements in capacitated NTU games, related those solutions to stable hypergraph matchings, and emphasized an important integrality consequence 
for \emph{normal} hypergraphs: A hypergraph $H$ is \emph{normal} if every sub-hypergraph $H'=(V, E')$ where $E'\subseteq E$ satisfies $\chi'(H')=\Delta(H')$, where $\chi'(H')$ is the chromatic index of $H'$, i.e., minimum number of matchings to cover $E'$ and $\Delta(H')$ is the maximum degree among all vertices in $H'$. 
Lov\'{a}sz \cite{Lovasz1972} showed that a hypergraph $H$ is \emph{normal} if and only if its associated fractional matching polytope $P(H)$ is integral. Consequently, every hypergraphic preference system on a normal hypergraph admits a stable matching. 

The above proof of existence of stable matching for normal hypergraphs does not give a polynomial-time algorithm. Scarf's pivoting procedure is finite, but no polynomial bound is known in general; in fact, computing a fractional stable matching of a hypergraphic preference system is PPAD-complete even under severe degree and rank restrictions \cite{KintaliEtAl2013, IshizukaKamiyama2018, csaji2022complexity}. 
Nevertheless, the above results motivate a central search question: 
\begin{quote}
    \textit{Does there exist a polynomial-time algorithm to find a stable matching in hypergraphic preference systems, where the hypergraph is a normal hypergraph?}
\end{quote}
The normal-hypergraph promise is algorithmically meaningful: Lov\'{a}sz \cite{Lovasz1972} characterized normal hypergraphs via the Helly property and perfection of the line graph, both of which are recognizable in polynomial time \cite{BergeDuchet1975, ChudnovskyEtAl2005}, thereby making normality testable in polynomial time. 

Recent work has answered this question for some structured subfamilies of normal hypergraphs and we discuss these now---see Figure \ref{fig:containment} for containment relations between these families. A \emph{subtree hypergraph} admits a tree on $V$ such that every hyperedge induces a subtree. Subtree hypergraphs are subfamilies of normal hypergraphs because subtrees of a tree have the Helly property and their intersection graphs are chordal \cite{Gavril1974}. Stable matchings on subtree hypergraphs can be found in polynomial time:  Cs\'{a}ji \cite{Csa-egres} obtained a polynomial-time algorithm for subtree hypergraphs by reducing stable matching to a kernel problem on clique-acyclic superorientation of a chordal graph---we will discuss this approach later; Bir\'{o}, Cs\'{a}ji, and Schlotter \cite{BCS25} subsequently gave a direct recursive algorithm. 

A different line of work studies \emph{unimodular hypergraphs}, namely hypergraphs whose vertex-hyperedge incidence matrix is totally unimodular \cite{BCS25}. Unimodular hypergraphs are normal via standard polyhedral results. Within the family of unimodular hypergraphs lies the family of \emph{network hypergraphs}---network hypergraphs are hypergraphs whose vertex-hyperedge incidence matrix is a \emph{network matrix}. Equivalently, its vertices correspond to the arcs of an oriented tree, known as the principal tree, and its hyperedges correspond to directed paths of the principal tree. 
\emph{Arborescence hypergraphs} are a proper subfamily of both network hypergraphs and subtree hypergraphs---we refer the reader to \cite{CFHS25} for the definition of arborescence hypergraphs. 
We will see later that hypergraphic preference systems over network hypergraphs contain bipartite stable matching as a special case (see Section \ref{sec:network-hypergraphs}). In contrast, bipartite stable matching cannot be modeled as a special case of stable matching over subtree hypergraphic preference systems. 
A suitable implementation of Scarf's algorithm is known to converge in polynomial time for arborescence hypergraphs \cite{CFHS25} as well as bipartite stable matching \cite{FaenzaHeSethuraman2025}, but those analyses do not seem to extend to network hypergraphs. Given this status, recent works identified network hypergraphs as the next natural frontier towards polynomial-time algorithms for stable matchings in hypergraphic preference systems over normal hypergraphs. Again, the network-hypergraph promise is algorithmically meaningful: there exists a polynomial-time algorithm to verify whether a given hypergraph is a network hypergraph (e.g., see \cite{schrijver-book-1}). In this work, we resolve the network hypergraph frontier by showing the existence of a polynomial-time algorithm. 

\begin{theorem}\label{thm:network-hypergraph-stable-matching-poly-time}
    There exists a polynomial-time algorithm to compute a stable matching of a given hypergraphic preference system, where the hypergraph is a network hypergraph.
\end{theorem}

The theorem can be viewed as an algorithmic strengthening of Bir\'{o} and Fleiner's \cite{BF16} existential argument. We emphasize that the result concerns with finding a stable matching; optimizing a weight function over stable matchings is NP-hard even for laminar hypergraphs which is a subfamily of network hypergraphs \cite{BCS25}. We refer the reader to Figure \ref{fig:containment} for the containment relation between the above-mentioned families of hypergraphs that are known to have a stable matching and the algorithmic status of computing a stable matching in each of them.

\begin{figure}[H]
\centering
\includegraphics[width=0.75\linewidth]{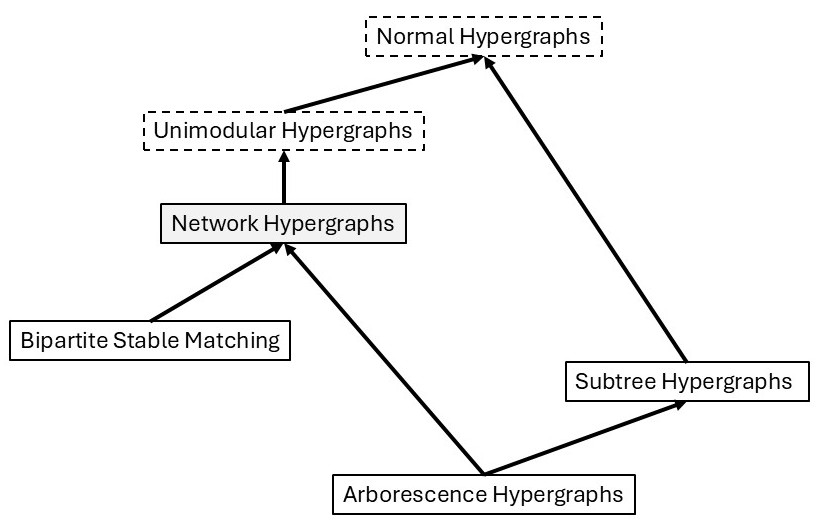}
\caption{Relevant hypergraph families for which the existence of stable matching is guaranteed and containment relations between them. An arrow from $A$ to $B$ indicates that $A$ is subfamily of $B$. The algorithmic status of hypergraph families within dashed boxes is still open while those within continuous boxes are known to be polynomial time. Our main contribution is resolving the algorithmic status of network hypergraphs. Network and subtree hypergraphs are incomparable. Bipartite stable matching is a special case of hypergraphic preference system over network hypergraphs but not subtree hypergraphs.}
\label{fig:containment}
\end{figure}

\paragraph{Proof outline.}
Our algorithm is based on a reduction to finding kernels in digraphs. For a hypergraphic preference system $(H, \succ)$, construct its conflict graph $G_H$ on node set $E(H)$ with a pair of hyperedge-nodes being adjacent if the hyperedges have non-empty intersection. Now, \emph{superorient} every edge of $G_H$ as follows: add an arc $e\rightarrow f$ if some vertex in $e\cap f$ prefers $f$ to $e$; both directions may be present when different common vertices have opposite preferences. Denote the resulting digraph by $D_{H, \succ}$, called the conflict preference graph. A \emph{kernel} $K$ in a digraph is a subset of vertices that forms a stable-set (i.e., no arc between any pair in $K$) and is absorbing (i.e., every vertex $u\not\in K$ has a vertex $v\in K$ such that $u\rightarrow v$ is present). Cs\'{a}ji \cite{Csa-egres} observed that stable matchings of $(H, \succ)$ are exactly the kernels of $D_{H, \succ}$. Essentially, the stable-set condition guarantees that the corresponding hyperedges form a matching, while absorption guarantees that every unchosen hyperedge points to a chosen one that defeats it at a common vertex thereby leading to stability. 

By the above reduction, it suffices to find the kernel of the digraph $D_{H, \succ}$ under the promise that $H$ is a network hypergraph. For this, we rely on the network representation of $H$. 
Each hyperedge of $H$ is a directed path in an oriented tree, and two hyperedges conflict exactly when their paths share a tree arc. Hence, $G_H$ is a \emph{directed-edge graph} in the terminology of Monma and Wei \cite{MW86}---abbreviated as \emph{DE graph}---i.e., an arc-intersection graph of directed paths in an oriented tree. Monma and Wei showed that directed paths in an oriented tree satisfy an arc-Helly property: every pairwise arc-intersecting family shares a common arc. From this, we infer that $D_{H, \succ}$ is a clique-acyclic super-orientation of the DE graph $G_H$, i.e., no clique of $G_H$ contains a unidirectional cycle in $D_{H, \succ}$. Pass-Lanneau, Igarashi, and Meunier \cite{PIM20} gave a polynomial-time algorithm for computing a kernel in a clique-acyclic super-orientation of a DE graph. Applying that algorithm and translating the kernel back to hyperedges gives the desired polynomial-time algorithm, thereby proving Theorem \ref{thm:network-hypergraph-stable-matching-poly-time}. 

Our argument parallels the earlier proof of polynomial-time stable matching in hypergraphic preference systems over subtree hypergraphs that appeared in the EGRES report of Cs\'{a}ji \cite{Csa-egres}. If $H$ is a subtree hypergraph, then the conflict graph $G_H$ is chordal and the Helly property of subtrees makes the conflict preference graph clique-acyclic; Pass-Lanneau, Igarashi, and Meunier's \cite{PIM20} polynomial-time kernel computation for clique-acyclic super-orientations of ``chordal'' graphs completes the proof for subtree hypergraphs. For network hypergraphs, the conflict graph $G_H$ need not be chordal (for an example, see the hypergraphic preference system in Figure \ref{fig:preference-system} which we will later see is a network hypergraph and its conflict graph $G_H$ in Figure \ref{fig:conflict-graph} which has a chordless cycle namely $C=\{f_1, f_2, f_3, f_4\}$), so that particular result of \cite{PIM20} cannot be invoked. The new ingredients are that the conflict graph is a DE graph, that directed paths satisfy the stronger arc-Helly property needed to certify clique-acyclicity, and that kernels are polynomial-time computable for clique-acyclic superorientations of DE graphs. We develop the kernel correspondence in Section \ref{sec:stable-matching-kernel}, the required facts about DE graphs in Section \ref{sec:tractable-kernels}, and exploit the tree representation of network hypergraphs to complete the proof of the main theorem in Section \ref{sec:network-hypergraphs}. Although the proof turned out to be simple via a sequence of known connections, its simplicity is apparent only in hindsight. We hope that the community will benefit from our presentation of these connections. 

%% file: stable-matching-and-kernel.tex
\section{Stable Hypergraph Matching and Digraph Kernel}\label{sec:stable-matching-kernel}
In this section, we define hypergraphic preference systems and stable matching of a hypergraphic preference system. Next, we show that there is a 1-to-1 correspondence between stable matchings of a hypergraphic preference system and kernels of a digraph derived from the hypergraphic preference system. This connection reduces the problem of finding a stable matching of a hypergraphic preference system to finding a kernel of a digraph derived from the hypergraphic preference system. The connection and the reduction were already observed by Cs\'{a}ji in his EGRES report \cite{Csa-egres}. We go through them to establish notation and terminology, and for the sake of completeness. 

\subsection{Hypergraphic Preference System}
We begin with a formal definition of a hypergraphic preference system. See Figure \ref{fig:preference-system} for an example of a hypergraphic preference system. 

\begin{definition}[Hypergraphic preference system]
    A hypergraph  $H=(V,E)$ is defined by a vertex set $V$ and a hyperedge set $E$, where every hyperedge $e\in E$ is a subset of $V$.
A \emph{hypergraphic preference system} is given by a pair $(H, \succ)$, where $H=(V,E)$ is a hypergraph and $\succ:=\{\succ_i:i\in V\}$ is the preference profile, $\succ_i$ being a strict order over $\delta(i)=\{e\in E:i\in e \}$ for each $i\in V$. For $e, e' \in \delta(i)$, we write $e\succeq_i e'$ if either $e\succ_i e'$ or $e=e'$. In addition, we assume that for every $i\in V$, the \emph{singleton hyperedge} $e_i=\{i\}\in\delta(i)$ and for every $e'\in\delta(i)$, $e'\succeq_i e_i$.\footnote{The assumption corresponds, in the bipartite setting, to the usual hypothesis that an agent prefers to be matched rather than being unmatched.}. 
\end{definition}

\begin{figure}
    \centering
    \includegraphics[width=0.8\linewidth]{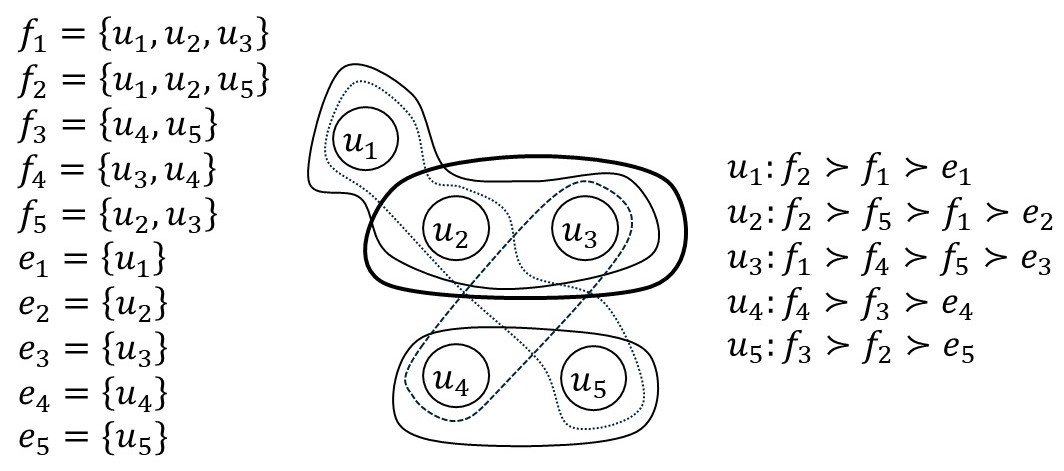}
    \caption{An example of a hypergraphic preference system. The hypergraph has vertex set $\{u_1, u_2, \ldots, u_5\}$, hyperedges $f_1, f_2, \ldots, f_5$, and singleton hyperedges $e_1, e_2, \ldots, e_5$. The preference profile of each vertex is shown on the right.}
    \label{fig:preference-system}
\end{figure}

Next, we define the notion of stable matching of a hypergraphic preference system. 

\begin{definition}[Stable matching]\label{def:sm}
    A \emph{stable matching} of a hypergraphic preference system $(H=(V, E), \succ)$ is a vector $x\in\{0,1\}^E$ so that for every hyperedge $e\in E$, there exists a vertex $v\in e$ such that
\begin{equation}\label{eq:sm}
    \sum_{e'\in\delta(v),e'\succeq_v e}x_{e'}=1.
\end{equation}
Equation~\eqref{eq:sm} for the self-loop hyperedge $e=e_v$ imposes that $x$ is the characteristic vector of a matching. Equation~\eqref{eq:sm} for an arbitrary hyperedge $e$ guarantees that there exists a vertex $v \in e$ and a hyperedge $e'$ in the matching (i.e.,~$x_{e'}=1$) so that $e'\succeq_v e$. 
\end{definition}

We encourage interested readers to find a stable matching for the hypergraphic preference system given in Figure \ref{fig:preference-system}---we will identify it in the next subsection. 
A vector $x\in [0, 1]^E$ satisfying Equation~\eqref{eq:sm} is said to be a \emph{fractional stable matching}. 
Aharoni and Fleiner \cite{AF03} showed that every hypergraphic preference system admits a fractional stable matching. However, their results is non-constructive. An interesting consequence of their result is that if the vertex-hyperedge incidence matrix of the hypergraph is totally unimodular, then every hypergraphic preference system associated with such a hypergraph has a stable matching (i.e., an integral stable matching). 

\subsection{Connection to Kernels of Digraphs}
In this section, we define kernels of a digraph and formalize a known encoding of stable matchings of a hypergraphic preference system to kernels of a digraph derived from the hypergraphic preference system. We begin with the definition of a kernel of a digraph. 
\begin{definition}
    Let $D$ be a digraph. 
    \begin{enumerate}
        \item A subset $K\subseteq V(D)$ is \emph{stable} if no two vertices $u, v\in K$ are adjacent in the underlying undirected graph of $D$. 
        \item A subset $K\subseteq V(D)$ is \emph{absorbing} if for every $u\in V(D)\setminus K$, there exists a vertex $k\in K$ such that $(u,k)\in \arcs(D)$. 
        \item A subset $K\subseteq V(D)$ is a \emph{kernel} if it is absorbing and stable.
    \end{enumerate}
\end{definition}
See Figure \ref{fig:kernel} for a digraph and its kernel. 
We emphasize that arbitrary digraphs may not have a kernel. Nevertheless, Cs\'{a}ji identified a nice encoding of stable matchings of a hypergraphic preference system and kernels of a digraph constructed from the hypergraphic preference system. We now state the construction of the digraph from the hypergraphic preference system and will subsequently, formalize the encoding. 

\begin{definition}
    Let $(H=(V, E), \succ)$ be a hypergraphic preference system. 
    \begin{enumerate}
        \item The \emph{conflict graph} $G_H$ of $H$ is an undirected simple graph on vertex set $E$ with two distinct vertices $e, f\in E$ being adjacent if $e\cap f\neq \emptyset$. 
        \item The \emph{conflict preference graph} $D_{H, \succ}$ is a digraph on vertex set $E$ with 
        \[\arcs\left(D_{H, \succ}\right):=\left\{(e,f): \exists v\in e\cap f\text{ such that } f\succ_v e\right\}.\]
    \end{enumerate}
\end{definition}

In the conflict preference graph, an arc is directed from a hyperedge $e$ to a conflicting hyperedge $f$ if $e$ is defeated by $f$ at some common vertex. Equivalently, for hyperedges $e, f\in E$ such that $e\cap f\neq \emptyset$, if all vertices in $e\cap f$ prefer $f$ to $e$ then only an arc from $e$ to $f$ is present; if all vertices in $e\cap f$ prefer $e$ to $f$ then only an arc from $f$ to $e$ is present; or if their rankings disagree on some pair of vertices in $e\cap f$, then arcs in both directions are present. See Figures \ref{fig:conflict-graph} and \ref{fig:conflict-preference-graph} for the conflict graph $G_H$ and the conflict preference graph $D_{H, \succ}$ of the hypergraphic preference system $(H, \succ)$ given in Figure \ref{fig:preference-system}. 

\begin{figure}[h]
    \centering
    \begin{subfigure}{.5\textwidth}
  \centering
  \includegraphics[width=.9\linewidth]{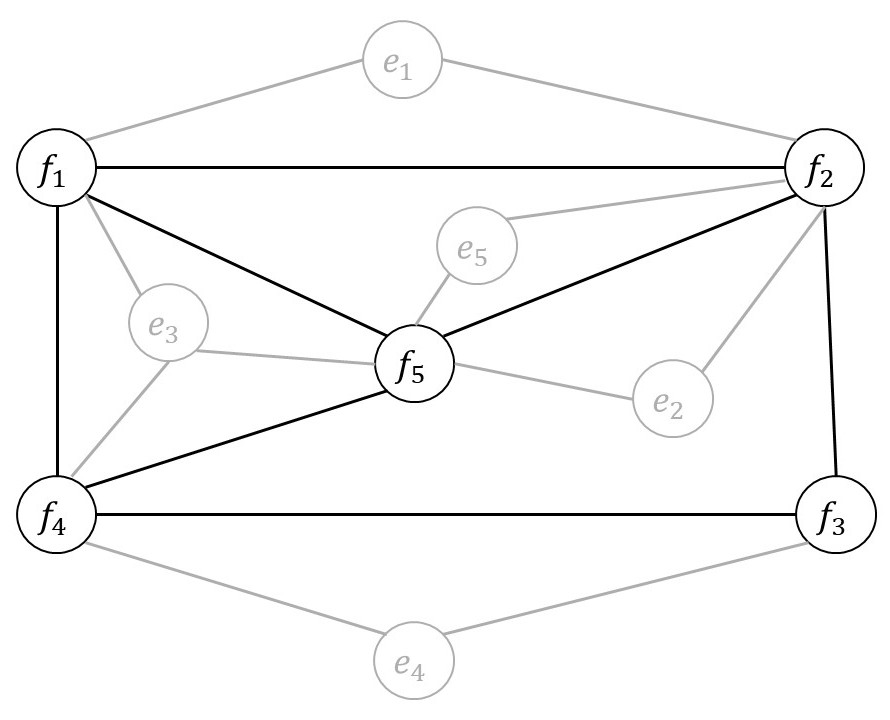}
  \caption{The conflict graph $G_H$. }
  \label{fig:conflict-graph}
\end{subfigure}%
    \begin{subfigure}{.5\textwidth}
  \centering
  \includegraphics[width=.9\linewidth]{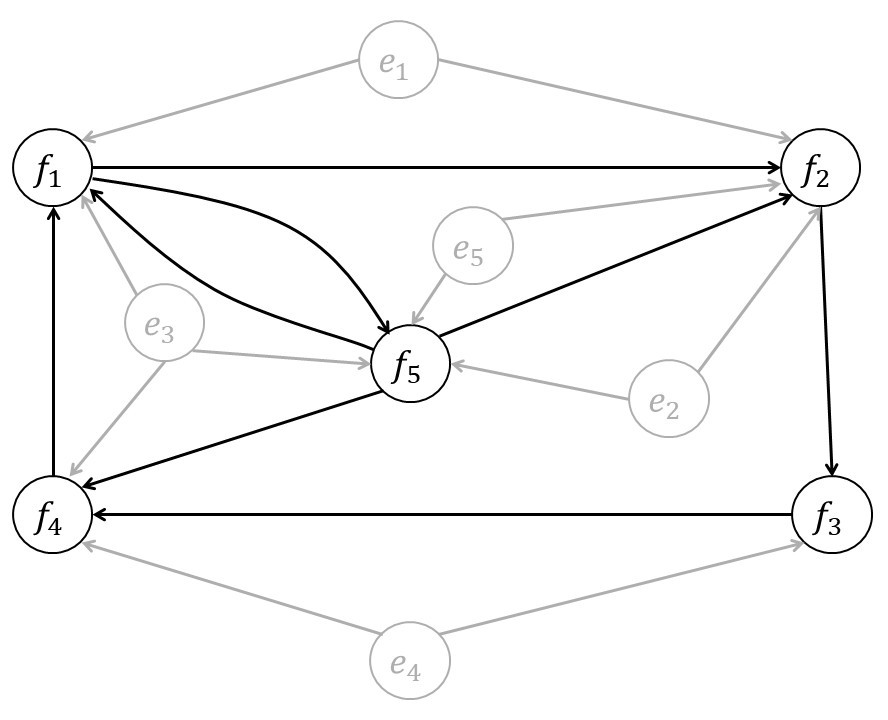}
  \caption{The conflict preference graph $D_{H,\succ}$.}
  \label{fig:conflict-preference-graph}
\end{subfigure}
\caption{The conflict graph $G_H$ and the conflict preference graph $D_{H, \succ}$ of the hypergraphic preference system from Figure \ref{fig:preference-system}. Singleton hyperedges are shown in gray. $K=\{f_2,f_4\}$ is a kernel of $D_{H, \succ}$.}
\label{fig:kernel}
\end{figure}

Cs\'{a}ji showed the following $1$-to-$1$ correspondence between stable matchings of a hypergraphic preference system and kernels of the conflict preference graph associated with the system. We include a proof for completeness. 
\begin{lemma}\label{lem:kernels-and-stable-matchings}\cite{Csa-egres}
    Let $(H=(V, E), \succ)$ be a hypergraphic preference system and $K\subseteq E$. The vector $\chi^{K}$ is a stable matching of $(H, \succ)$ if and only if $K$ is a kernel of $D_{H, \succ}$. 
\end{lemma}
\begin{proof}
    $(\Leftarrow)$ Suppose that $K\subseteq E$ is a kernel of $D_{H,\succ}$. We show that $\chi^{K} \in \left\{0,1\right\}^E$ is a stable matching of $(H,\succ)$. Let $e\in E$. We need to show that there exists a vertex $v\in e$ such that $\sum_{e'\in \delta(v), e'\succ_v e}\chi^K_{e'}= 1$. We case on whether $e\in K$ or $e\in E\setminus K$. 
    
    Case (i). Suppose $e \in E \setminus K$. By definition of kernel, there exists a hyperedge $f \in K$ s.t. $(e,f) \in \text{Arcs}(D_{H,\succ})$. This implies that there exists a vertex $v \in e \cap f$ s.t. $f \succ_v e$. Thus, $\sum_{e' \in \delta(v),e' \succeq_v e}\chi^{K}_{e'} \geq 1$. Next, we show equality for this vertex $v$.  
    Suppose there exists a hyperedge $g \in \delta(v) \cap K$ s.t., $g \succeq_v e$, $g \neq f$. Then since $\succ_v$ is a total ordering of the hyperedges in $\delta(v)$, either $g \succ_v f$ or $f \succ_v g$. Consequently, either $(g,f)$ or $(f,g)$ is an arc of $D_{H,\succ}$, which contradicts the fact that $K$ is stable. Thus, $\sum_{e' \in \delta(v),e' \succeq_v e}\chi^{K}_{e'} = 1$. 
    
    Case (ii). Suppose $e \in K$. Let $v \in V$ be an arbitrary vertex in $e$. There are no edges $e' \in K$ s.t. $e' \succ_v e$ since otherwise $e$ and $e'$ would be adjacent in $D_{H,\succ}$, which would contradict the stability of $K$. Thus, $\sum_{e' \in \delta(i),e' \succeq_v e}\chi^{K}_{e'} = 1$.
    
    $(\Rightarrow)$ Assume $\chi^{K}$ is a stable matching for $(H,\succ)$. Fix a hyperedge $e \in E \setminus K$. Since $\chi^{K}$ is a stable matching, there a hyperedge $f \in K$ s.t. $f \succ_i e$ for some $i\in e\cap f$. Thus, $(e,f) \in \text{Arcs}(D_{H,\succ})$ which implies that $K$ is absorbing. If $K$ was not stable, then there exist hyperedges $e_1,e_2 \in K$ s.t. $e_1$ and $e_2$ are adjacent in $D_{H,\succ}$. But this implies that there exists a vertex $i \in V$ s.t. $\sum_{e' \in \delta(i),e' \succeq_i e}\chi^{K}_{e'} \geq 2$ for some $e\in \{e_1, e_2\}$, which contradicts the fact that $\chi^{K}$ is a stable matching.
\end{proof}
By Lemma \ref{lem:kernels-and-stable-matchings}, $\{f_2, f_4\}$ is a stable matching for the hypergraphic preference system given in Figure \ref{fig:preference-system}.

%% file: DE-graphs.tex
\section{Tractable Digraph Families for Kernel Computation}\label{sec:tractable-kernels}

Pass-Lanneau, Igarashi, and Meunier \cite{PIM20} identified certain families of digraphs for which kernel computation is polynomial-time tractable. 
For the purposes of stable matching in network hypergraphs, we rely on kernel computability in a clique-acyclic super-orientation of a DE graph. We define these concepts and state their result next. 

A directed graph ${T}$ is a \emph{oriented tree} if the underlying undirected graph (obtained by dropping the orientation of all arcs) is connected and acyclic. For a digraph $D$, a \emph{$D$-directed path} is a directed path in $D$. 
\begin{definition}
    A \emph{directed-edge graph}, abbreviated \emph{DE graph}, 
    is an undirected graph $G=(V, E)$ such that there exists an oriented tree $T$ with the property that every vertex $v\in V$ corresponds to a $T$-directed path and two vertices are adjacent if the corresponding paths share an arc. 
\end{definition}
Equivalently, a DE graph is the path-arc intersection graph in an oriented tree $T$. 
We emphasize that DE graphs are simple undirected graphs. DE graphs were introduced and characterized by Monma and Wei \cite{MW86}. Their characterization implies a polynomial-time algorithm to verify whether a given graph is a DE graph and if so, then compute a representation, i.e., an oriented tree $T$ and a collection $\mathcal{P}$ of directed paths in $T$ corresponding to the vertices of $G$. 

Next, we define the notion of super-orientation of a simple undirected graph. 

\begin{definition}
A \emph{super-orientation} of a simple undirected graph $G$ is a digraph obtained by replacing every edge $\{u, v\}$ of $G$ by at least one of the arcs $(u, v)$ and $(v, u)$ (both arcs are allowed).
\end{definition}

We will be interested in special kinds of super-orientations of a simple undirected graph, namely clique-acyclic super-orientations. 

\begin{definition}
A super-orientation $D$ of a simple undirected graph $G$ is \emph{clique-acyclic} if no clique of $G$ has a directed cycle in $D$ consisting of arcs oriented in only one direction (i.e., cliques of $G$ do not contain unidirectional cycles in $D$). 
\end{definition}

Clique-acyclic super-orientation of a simple undirected graph can equivalently be characterized via the presence of sinks in every clique. For a digraph $D$ and a subset $U\subseteq V(D)$, a vertex $s$ is a \emph{$U$-sink in $D$} if $(u,s)\in \arcs(D)$ for every $u\in U\setminus s$. 

\begin{proposition}\label{prop:clique-acyclic-characterization}\cite{PIM20}
    Let $G$ be an undirected graph and $D$ be a super-orientation of $G$. 
    $D$ is clique-acyclic if and only if every clique $Q$ of $G$ has a $V(Q)$-sink in $D$. 
\end{proposition}
Verifying clique-acyclicity of a given super-orientation $D$ of an arbitrary undirected graph $G$ is co-NP-complete (see Corollary 11 in \cite{AH15}). Clique-acyclicity of a given super-orientation of a DE graph can be verified in polynomial-time: It is easy to see that $D$ is clique-acyclic if and only if $D'[Q]$ is acyclic for every maximal clique $Q$ of $G$, where $D':=(V(G), \{(u,v): (u,v)\in \arcs(D), (v,u)\not\in \arcs(D)\}$; Monma and Wei \cite{MW86} showed that the number of maximal cliques in a DE graph is $O(n)$ and they can all be enumerated in polynomial time.

Pass-Lanneau, Igarashi, and Meunier \cite{PIM20} showed that there exists a polynomial-time algorithm to compute a kernel of a given clique-acyclic super-orientation of a DE graph. 

\begin{theorem}\cite{PIM20}\label{thm:kernel-of-de-graphs}
    Given a digraph $D$ such that $D$ is a clique-acyclic super-orientation of a DE graph $G$, 
    there exists a polynomial-time algorithm to compute a kernel of $D$. 
\end{theorem}

%% file: network-hypergraphs.tex
\section{Stable Matching in Network Hypergraphs}\label{sec:network-hypergraphs}
We prove the main theorem in this section. We begin with the definition of network hypergraphs. 
Let $T = (U, A_0)$ be an oriented tree and let $(s,t)$ be an ordered pair of vertices of $T$. There exists a unique undirected path between $s$ and $t$ in $T$. We view this path as a directed traversal from $s$ to $t$ and call an arc $a_0\in A_0$ to be forward on this traversal if $a_0$ is traversed from its tail to its head and backward otherwise. 

\begin{definition}[Network Matrix]\label{def:network-matrix}
Let ${D}=(U,{A})$ be a directed graph and ${T}=(U,{A}_0)$ be a directed tree on vertex set $U$. The network matrix corresponding to $({D}, {T})$ is the matrix $M\in\{0,\pm 1\}^{{A}_0\times {A}}$ where for every $a=(u,v)\in{A}$ and $a_0\in {A}_0$, we have 
\begin{equation}\label{eq:Network-Matrix}
    M_{a_0,a}=\left\{
    \begin{array}{cc}
        1 & \textrm{if $a_0$ is a forward arc on the unique ${T}$-path from $u$ to $v$;}\\
        -1 & \textrm{if $a_0$ is a backward arc on the unique ${T}$-path from $u$ to $v$;} \\
        0 & \textrm{if $a_0$ does not belong to the unique ${T}$-path from $u$ to $v$.}
    \end{array}
    \right.
\end{equation}

A matrix $M$ is a \emph{network matrix} if there exists a directed graph ${D}=(U,{A})$ and a directed tree ${T}=(U,{A}_0)$ such that the network matrix corresponding to $({D}, {T})$ is $M$. We say ${T}$ is the \emph{principal tree} corresponding to $M$.
\end{definition}

\begin{definition}[Network Hypergraph]\label{def:Network_HG}
    A hypergraph $H=(V,E)$ is a \emph{network hypergraph} if the node-hyperedge incidence matrix of $H$ is a network matrix. If the node-hyperedge incidence matrix of $H$ is the network matrix corresponding to $(D_H,T_H)$, then we call $T_H$ as the \emph{principal tree} of $H$ and $(D_H,T_H)$ as the \emph{underlying network} of $H$. 
\end{definition}

The hypergraph in Figure \ref{fig:preference-system} is a network hypergraph---see Figure \ref{fig:principal-tree} for a principal tree of this hypergraph. 
\begin{figure}
    \centering
\includegraphics[width=0.3\linewidth]{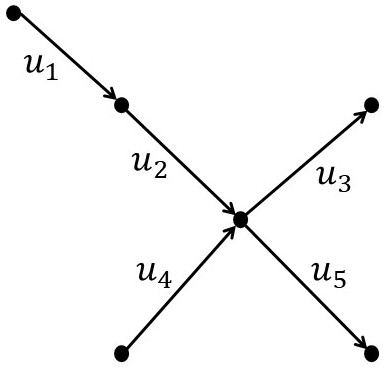}
\caption{A principal tree of the hypergraph in Figure \ref{fig:preference-system}. Note that every hyperedge is a directed path in this tree.}
\label{fig:principal-tree}
\end{figure}
For a network hypergraph $H$, the underlying network $(D_H, T_H)$ is not necessarily unique. Nevertheless, there exists a polynomial-time algorithm to verify whether a given hypergraph is a network hypergraph and, if so, then compute an underlying network of $H$ (e.g., see Schrijver's book \cite{schrijver-book-1}). 

Bipartite stable matching is a special case of stable matching on hypergraphic preference systems, where the hypergraph is a network hypergraph: orient the arcs representing one side of the bipartition toward the center of a star and those representing the other side away from it, so that every acceptable pair becomes a directed
two-arc path.

The next proposition summarizes an alternative viewpoint of a network hypergraph, namely that it is the path-arc intersection graph in the principal tree.

\begin{proposition}\label{prop:network-hypergraph-to-directed-paths-in-the-tree}
    Let $H=(V, E)$ be a network hypergraph and let $T_H$ be a principal tree of $H$. The arc set of $T_H$ is $V$. Moreover, for every hyperedge $e\in E$, there exists a $T_H$-directed path $P_e$ whose arc set is exactly $e$. 
\end{proposition}

A useful consequence of Proposition \ref{prop:network-hypergraph-to-directed-paths-in-the-tree} is that the conflict graph $G_H$ associated with a network hypergraph $H$ is a DE graph. 

\begin{corollary}\label{coro:network-hypergraphs-DE-graphs}
    Let $H$ be a network hypergraph. Then, $G_H$ is a DE graph.
\end{corollary}
\begin{proof}
    Let $T_H$ be a principal tree of $H$. 
    By Proposition \ref{prop:network-hypergraph-to-directed-paths-in-the-tree}, each $e\in E$ is represented by a $T_H$-directed path $P_e$. We also have that $e\cap f\neq \emptyset$ if and only if $P_e$ and $P_f$ share a tree arc. Thus, the conflict graph $G_H$ is a DE graph. 
\end{proof}

For the rest of this section, consider a hypergraphic preference system $(H, \succ)$, where $H$ is a network hypergraph. We will show that the conflict preference graph $D_{H, \succ}$ is indeed a clique-acyclic super-orientation of the conflict graph $G_H$. We begin with the observation that the conflict preference graph is a super-orientation of the conflict graph. This was observed by Cs\'{a}ji already. We include a proof for completeness.

\begin{proposition}\label{prop:D-is-superorientation}\cite{Csa-egres}
    Let $(H, \succ)$ be a hypergraphic preference system. Then, $D_{H, \succ}$ is a super-orientation of the conflict graph $G_H$. 
\end{proposition}
\begin{proof}
    Let $e \in E(G_H)$. By definition of $E(G_H)$, there exists a vertex $v \in f \cap g$. Since $\succ_{v}$ is a total ordering of the hyperedges in $\delta(v)$, either $f \succ_{v} g$ or $g \succ_{v} f$. Thus either $(f,g) \in \text{Arcs}(D_{H,\succ})$ or $(g,f) \in \text{Arcs}(D_{H,\succ})$.
\end{proof}

Next, we show that the conflict preference graph is a clique-acyclic super-orientation of the conflict graph. 
For this, we need the following directed-arc Helly property of oriented trees which was shown by Monma and Wei. 
We include a proof for completeness. Their proof was by induction on the number of paths while our proof is via slightly different arguments. 

\begin{lemma}\cite{MW86}\label{lem:directed-arc-helly}
    Let $T$ be an oriented tree and let $\Pi$ be a finite family of non-empty $T$-directed paths. If every two paths in $\Pi$ share an arc, then there exists an arc shared by all paths in $\Pi$. 
\end{lemma}
\begin{proof}
    If $|\Pi|=1$, then we are done. So, assume $|\Pi|\ge 2$. 
    
    We first show that there exists a node $x\in \cap_{P\in \Pi} V(P)$. Let $T'$ be the undirected tree obtained from $T$ by ingoring its arc orientations. For each $P\in \Pi$, the vertex set $V(P)$ induces a subtree of $T'$. The hypothesis that every two members of $\Pi$ share an arc implies, in particular, that the subtree $T'[V(P)]$ are pairwise vertex-intersecting. By the Helly property for subtrees of a tree, it follows that there exists a node $x\in \cap_{P\in \Pi} V(P)$. 

    Fix a node $x\in \cap_{P\in \Pi}V(P)$. For $P\in \Pi$, let $I_P:=\arcs(P)\cap \delta_T(x)$ be the set of arcs of $P$ incident with $x$. Since $P$ is non-empty and contains $x$, the set $I_P$ is non-empty. Moreover, because $P$ is a directed path, $I_P$ contains at most one arc entering $x$ and at most one arc leaving $x$. Thus, $|I_P|\le 2$ and whenever $|I_P|=2$, one member of $I_P$ lies in $\delta_T^{in}(x)$ and the other lies in $\delta_T^{out}(x)$. 

    We next show that the collection $\{I_P: P\in \Pi\}$ is pairwise intersecting. Let $P, Q\in \Pi$ and let $a$ be an arc shared by $P$ and $Q$. The intersection of two paths in a tree is connected. Consequently, the intersection of the underlying paths of $P$ and $Q$ contains the unique path between $x$ and either end-vertex of $a$. Choosing an end-vertex different from $x$ if $a$ is incident with $x$, the first arc on this path is incident with $x$ and belongs to both $P$ and $Q$. Hence, $I_P\cap I_Q\neq \emptyset$. 

    We now show that the whole family $\{I_P: P\in \Pi\}$ has a common member which would complete the proof. If $I_{P_0}=\{a\}$ for some $P_0$, then pairwise intersection forces $a\in I_P$ for every $P\in \Pi$ and we are done. We may therefore assume that every $I_P$ has size two. In particular, we may regard each $I_P$ as an edge of the bipartite graph with bipartition $(\delta_T^{in}(x), \delta_T^{out}(x))$. Fix an arbitrary $P_0\in \Pi$ and let $I_{P_0}=\{a, b\}$ with $a\in \delta_T^{in}(x)$ and $b\in \delta_T^{out}(x)$. For every $P\in \Pi$, we have that $I_P$ contains $a$ or $b$. Suppose there exists $I_P$ that misses $a$ and some $I_Q$ that misses $b$, then pairwise intersection with $I_{P_0}$ gives $I_P=\{a', b\}$ and $I_Q=\{a, b'\}$ for some $a'\neq a$ and some $b'\neq b$ and moreover, $a, a'\in \delta_T^{in}(x)$ and $b, b'\in \delta_T^{out}(x)$. This implies that $I_P\cap I_Q=\emptyset$, a contradiction to the pairwise intersecting property. Therefore, either every $I_P$ contains $a$ or every $I_P$ contains $b$. That common arc is an arc of $T$ belonging to every path in $\Pi$. 
\end{proof}

We emphasize that the undirected version of Lemma \ref{lem:directed-arc-helly} is false: suppose $T$ is an undirected tree and $\Pi$ is a finite family of non-empty $T$-paths (undirected paths in $T$). If every two paths in $\Pi$ share an edge, then there may not be an edge shared by all paths in $\Pi$. Here is a simple example: consider the tree $T$ to be the star graph with center vertex $v$ and $3$ leaf vertices denoted as $\{a, b, c\}$ and the collection $\Pi$ of paths connecting the three pairs of leaves. These three paths are pairwise intersecting but they have no common edge. 

We now use Lemma \ref{lem:directed-arc-helly} to conclude that the conflict preference graph $D_{H, \succ}$ is a clique-acyclic super-orientation of the conflict graph $G_H$.

\begin{lemma}\label{lem:network-hypergraph-to-clique-acyclic-super-orientation}
    Let $(H, \succ)$ be a hypergraphic preference system, where $H$ is a network hypergraph. Then, 
    $D_{H, \succ}$ is a clique-acyclic super-orientation of the conflict graph $G_H$. 
\end{lemma}
\begin{proof}
    By Proposition \ref{prop:D-is-superorientation}, we know that $D_{H, \succ}$ is a super-orientation of the conflict graph $G_H$. It remains to show clique-acyclicity. Let $Q$ be an arbitrary clique of $G_H$. By Proposition \ref{prop:clique-acyclic-characterization}, it suffices to show that there exists a $V(Q)$-sink in $D_{H, \succ}$. We note that the vertices of $G_H$ correspond to hyperedges of $H$ and hence, the vertices of the clique $Q$ correspond to a subset of hyperedges that are pairwise-intersecting. 

    Let $H=(V, E)$ and let $T_H$ be a principal tree of $H$. We recall that the arc set of $T_H$ is $V$.  
    By Proposition \ref{prop:network-hypergraph-to-directed-paths-in-the-tree}, each $e\in E$ is represented by a $T_H$-directed path $P_e$. In particular, the collection $\Pi:=\{P_e: e\in Q\}$ of $T$-directed paths are pairwise arc-intersecting (since the vertices of $Q$ correspond to a subset of pairwise intersecting hyperedges). By Lemma \ref{lem:directed-arc-helly}, there exists a common arc $v$ in $T_H$ shared by all paths in $\Pi$. Let $f$ be a $\succ_v$-maximum member of $Q$, i.e., $f\in V(Q)$ and $f\succ_v e$ for all $e\in V(Q)\setminus \{f\}$. It follows that for every $e\in V(Q)\setminus \{f\}$, the arc $(e, f)$ is present in $D_{H, \succ}$ by definition of $D_{H, \succ}$ and consequently, $f$ is the needed $V(Q)$-sink in $D_{H, \succ}$. 
\end{proof}

We now complete the proof of Theorem \ref{thm:network-hypergraph-stable-matching-poly-time} using Lemmas \ref{lem:kernels-and-stable-matchings} and \ref{lem:network-hypergraph-to-clique-acyclic-super-orientation}, Corollary \ref{coro:network-hypergraphs-DE-graphs}, and Theorem \ref{thm:kernel-of-de-graphs}.
\begin{proof}[Proof of Theorem \ref{thm:network-hypergraph-stable-matching-poly-time}]
    We first describe the algorithm. Let $(H=(V, E), \succ)$ be a given hypergraphic preference system where $H$ is a network hypergraph. Construct the conflict graph $G_H$ of $H$ and the digraph $D_{H, \succ}$. Compute a kernel $K$ of $D_{H, \succ}$ and return it. 
    
    The correctness of the algorithm is by Lemma \ref{lem:kernels-and-stable-matchings}. We now show that the algorithm can be implemented to run in polynomial time. Both $G_H$ and $D_{H, \succ}$ can be constructed in polynomial time. 
    By Corollary \ref{coro:network-hypergraphs-DE-graphs}, the conflict graph $G_H$ is a DE graph. 
    By Lemma \ref{lem:network-hypergraph-to-clique-acyclic-super-orientation}, the digraph $D_{H, \succ}$ is a clique-acyclic super-orientation of $G_H$. Consequently, by Theorem \ref{thm:kernel-of-de-graphs}, a kernel of $D_{H, \succ}$ can be computed in polynomial time. Thus, the algorithm can be implemented to run in polynomial time. 
\end{proof}

%% file: conclusion.tex
\section{Future Directions}
Our result advances the algorithmic frontier for stable matching in hypergraphic preference systems from arborescence hypergraphs to network hypergraphs (see Figure \ref{fig:containment}). In contrast to the subtree hypergraphs route towards normal hypergraphs, network hypergraphs include bipartite stable matching as a special case. 
Two natural questions mark the boundary of our result. It is unknown whether Scarf's pivoting procedure itself converges in polynomial time on network hypergraphs, and whether stable matching is polynomial-time solvable on all unimodular hypergraphs. Beyond them lies the central problem regarding normal hypergraphs. 

Stable hypergraph matching for arbitrary capacities is a useful model in certain applications \cite{BCS25}. 
In stable hypergraph $b$-matching, a vertex $v$ may belong to up to $b(v)$ selected hyperedges. Scarf's lemma extends to this setting, and total unimodularity makes the relevant capacitated matching polytope integral; thus, a stable $b$-matching exists on every unimodular hypergraph \cite{BF16}. The algorithmic picture is incomplete. Polynomial-time solvability is known for laminar hypergraphs, and an XP algorithm parameterized by maximum hyperedge size is known for sub-path hypergraphs, while existence is NP-hard for subtree hypergraphs under strong restrictions \cite{BCS25}. Our approaches do not extend to stable $b$-matching since our reduction to digraph kernel fundamentally uses capacity one: kernels encode stable sets, whereas a $b$-matching may contain intersecting hyperedges. Our work, therefore, raises polynomial-time stable $b$-matching for network hypergraphs as another natural direction for future research.